\documentclass[conference]{IEEEtran}

\normalsize

\usepackage[nolist,nohyperlinks]{acronym}
\usepackage[caption=false]{subfig}

\usepackage{graphicx}

\usepackage{multirow}
\usepackage{cite}
\usepackage{empheq}
\usepackage[mathcal]{eucal}

\usepackage{amsmath}

\usepackage{amsfonts}
\usepackage{amssymb}

\usepackage{algorithm}
\usepackage{algorithmic}
\usepackage{pgf}
\usepackage{pgfplots}
\usepackage{pgfplotstable}
\pgfplotsset{filter discard warning=false}
\usetikzlibrary{decorations.fractals, shapes.misc, backgrounds, arrows, decorations, positioning, shapes.symbols, calc, shapes, fit, shadows, patterns}
\usepackage{bytefield}
\newlength\figureheight
\newlength\figurewidth
\newlength\rv
\usepackage{amsmath}

\usepackage{multirow}
\usepackage{amsthm}

\newtheorem{theorem}{Theorem}
\newtheorem{corollary}[theorem]{Corollary}

\begin{document}

\title{A Markov Chain Model for the Decoding Probability of Sparse Network Coding}


\author{Pablo Garrido$^\dag$, Daniel E. Lucani$^\$$, Ram\'on Ag\"uero$^\dag$\\[6pt]
	\begin{minipage}{0.5\textwidth}
		\begin{center}
			$^\dag$ University of Cantabria\\
			Santander, Spain\\
			\texttt{\small \{pgarrido,ramon\}@tlmat.unican.es}
		\end{center}
	\end{minipage}
	\begin{minipage}{0.5\textwidth}
		\begin{center}
			$^\$$ Department of Electronic Systems \\
			 Aalborg University, Denmark\\
			\texttt{\small del@es.aau.dk}
		\end{center}
	\end{minipage}
}

\maketitle           

\begin{abstract}
Random Linear Network Coding (RLNC) has been proved to offer an efficient communication scheme, leveraging an interesting robustness against packet losses. However, it suffers from a high computational complexity and some novel approaches, which follow the same idea, have been recently proposed. One of such solutions is Tunable Sparse Network Coding (TSNC), where only few packets are combined in each transmissions. The amount of data packets to be combined in each transmissions can be set from a density parameter/distribution, which could be eventually adapted. In this work we present an analytical model that captures the performance of SNC on an accurate way. We exploit an absorbing Markov process where the states are defined by the number of useful packets received by the decoder, i.e the decoding matrix rank, and the number of non-zero columns at such matrix. The model is validated by means of a thorough simulation campaign, and the difference between model and simulation is negligible. A mean square error less than $4 \cdot 10^{-4}$ in the worst cases. We also include in the comparison some of more general bounds that have been recently used, showing that their accuracy is rather poor. The proposed model would enable a more precise assessment of the behavior of sparse network coding techniques. The last results show that the proposed analytical model can be exploited by the TSNC techniques in order to select by the encoder the best density as the transmission evolves.

\end{abstract}

\section{Introduction}
\label{Sec: Introduction}

\ac{NC} techniques foster a new communication paradigm, where packets are no longer immutable entities and nodes across the network could retransmit, discard or recode the packets. Among these techniques, \ac{RLNC} stands as one of the most interesting solutions, since it provides robustness against packet losses. On the other hand, some questions have been raised about the decoding complexity of \ac{RLNC}. 

In order to reduce such complexity, \ac{SNC} techniques came out. In particular, a \ac{TSNC} scheme was proposed by Feizi~\emph{et al.}~\cite{Feizi2012}; in a nutshell they proposed tuning the density of the coded packets, as they are being generated by the source during the transmission. There is a trade-off between the reduction of the computational complexity and the performance degradation induced by the corresponding overhead. Hence, it would be really helpful if we could find the optimum density configuration. However, there was not an appropriate model for sparse coding techniques, and only some theoretical bounds have been used. These bounds are trying to apply in a larger number of cases and have focused on a single dimension: the degrees of freedom, i.e decoding matrix rank, as part of the relevant information. As will be seen later, their accuracy is quite poor. 

In this paper we propose a complete analytical model for \ac{SNC} techniques. We include a second dimension, the covered packets, i.e the non-zero columns at the decoding matrix. It is based on an Absorbing Markov Chain and it precisely mimics the probability of generating new information when sparse coding schemes are used. To our best knowledge, there were not any similar proposal in the related literature. We shall later see that the accuracy of the model is very high.

The rest of the paper is structured as follows: Section~\ref{Sec: Related} summarizes the operation principles of \ac{NC} techniques and recalls some of the bounds that have been used in previous works to estimate the performance of \ac{SNC}. In Section~\ref{Sec: Implementation} we describe the proposed model and we exploit the properties of Absorbing Markov Chains to assess the performance of \ac{SNC} techniques. Finally, in~\ref{Sec: Results} we validate our proposal, by means of an extensive simulation campaign; we also compare the performance with that exhibited by some of the bounds that have been previously used in the related literature. We finally conclude the paper in~\ref{Sec: Conclusions}, highlighting its most relevant contributions and advocating some aspects that will be addressed in our future work, by exploiting the proposed model.

\section{From Random Linear Network Coding to Tunable Sparse Network Coding}
\label{Sec: Related}

\subsection{RLNC}

\ac{NC} techniques were originally proposed by \emph{Ahlswede et al.} in their seminal paper~\cite{Ahlswede2000}, where the \emph{store and forward} approach was questioned; they also proved that the use of a coding scheme yields the maximum multicast capacity. Some subsequent works by Koetter and Li~\cite{Koetter:2003,Li2003} broadened that idea by using linear codes, and Ho~\emph{et al.}~\cite{Ho2003} presented a randomized network coding approach that achieves the maximum multicast capacity with high probability, advocating the \ac{RLNC} scheme. Since those initial works, we have witnessed an increasing interest on potential applications of \ac{NC}.

Many works have studied the benefits of \ac{NC} techniques; Katty~\cite{Katti2008} and Chachulski~\cite{Chachulski2007} were some of the first ones proposing actual protocols, COPE and \ac{MORE}, respectively. Each of them represents a different \ac{NC} flavor: \emph{Inter-flow} and \emph{Intra-flow}. In \emph{Inter-flow} \ac{NC} packets belonging to different information flows are combined. Although this approach has been thoroughly analyzed~\cite{Katti2008, Zhao2010, Zhao2012, traskov2006}, they exhibit some drawbacks if applied over realistic scenarios, as was shown in~\cite{Cloud2011, GomezPIMRC2012}. 

On the other hand, \emph{Intra-flow} \ac{NC} techniques are based on the combination of packets belonging to the same flow; among them, the \ac{RLNC} scheme stands out as the most widespread solution, due to its simplicity and good performance. Indeed, it hides losses from the upper layers over point-to-point links~\cite{Sundararajan2009, GomezISSC2014}, reduces signalling overhead over opportunistic networks~\cite{Chachulski2007}, and leverages efficient transmissions over wireless mesh networks~\cite{Pahlevani2013, GomezWD2014, Pandi2015}. 


There are also other coding solutions, LT~\cite{Luby2012} and Raptor Codes~\cite{Shokrollahi2006}, having some of the advantages of the \ac{RLNC} scheme: (i) resiliency to packet losses, (ii) low overhead and (iii) suitability for heterogeneous networks and devices. However, they do not provide (i) on-the-fly decoding, (ii) low delay, and (iii) recoding capabilities, which are considered to be some of the greatest advantages of \ac{NC}. 


On the other hand, the main argument questioning the use of \ac{RLNC} is their decoding complexity, $\mathbb{O}(k^3)$, where $k$ is the number of packets to decode, which is considerably higher than other approaches (for instance, LT or Raptor Codes). The coding throughput, defined as the rate at which coding is carried out, is severely impacted by the coding parameters (Galois Field size, $GF(2^q)$, and generation size, $k$). In general, greater values of these parameters would lead to lower coding throughput~\cite{Shojania2007,Sorensen2016}. Furthermore, network overhead, which is mainly due to the transmission of useless packets (linear dependent combinations) and the corresponding protocol header, is also affected by the coding parameters, as Heide~\emph{et al.} analyzed in~\cite{Heide2013}, focusing on the field and generation sizes, and the code sparsity.  

Several works have focused on reducing the coding and decoding complexity proposing different alternatives. One initial  idea to separate the \ac{RLNC} performance from the amount of transmitted data, was to divide the data into generations~\cite{Chou2003}, chunks~\cite{Maymounkov2006}, segments~\cite{Wang2007} or groups~\cite{Gkantsidis2006} of $k$ packets. A more advanced scheme~\cite{Sorensen2015} proposed the use of overlapping generations. More recent proposals exploit sparse coding schemes to reduce the complexity, following the approach of LT or Raptor Codes. For instance, Heide~\emph{et al.} presented a sparse coding scheme with a non-random selection, which would allow different recoding mechanisms. Another example was \ac{TSNC} scheme, by Feizi~\emph{et al.}~\cite{Feizi2012}, which fosters a dynamic increase of the coding density as long as the transmission goes on. 

Feizi's work was later broadened in~\cite{Feizi2014}, where different recoding approaches and more complex tuning functions were proposed. Additionally, the authors of~\cite{Sorensen2015} proposed a practical implementation of the \ac{TSNC} approach. They used a lower bound to estimate the impact of the density on the overall overhead, but as will be seen later, such bound is not very accurate.

Another approaches to reduce both the complexity and the overhead advocate the use of inner and outer codes, such as Fulcrum Codes~\cite{Lucani2014} or \ac{BATS} Codes~\cite{Yang2014}. Fullcrum Codes, which are suitable for heterogeneous devices, foster using different field sizes to decode the received packets. \ac{BATS} Codes reduce the computational complexity by means of an outer code based on a Fountain coding scheme.

\subsection{Tunable Sparse Network Coding}

As already mentioned, Tunable Sparse Network Coding was initially proposed by Feizi~\emph{et al.}~\cite{Feizi2012}, and was later broadened in~\cite{Feizi2014, Sorensen2015}. The main reasoning was that the complexity of the decoding process for traditional \ac{RLNC} solutions is considerably higher than in other approaches, for instance LT or Raptor Codes, which exploit sparse coding techniques. LT or Raptor Codes propose a random distribution of the density, where sparse packets  (built by the combination of a few original packets) are more likely to be sent. 

On the other hand, the legacy \ac{RLNC} scheme generates coded packets by a random selection of coefficients $c_i$, selected from a Galois Field, $GF(2^q)$, for each of the original $k$ packets of a generation:

\begin{equation}
 p' = \sum_{i = 0}^{k}c_i \times p_i
\end{equation}

Using a Sparse Coding Scheme only a set of $w$ randomly selected packets, $W = \{p_{i1}, p_{i2}, \cdots, p_{iw} | p_{ik} \ne p_{ik'}, \forall i_k \ne i_{k'} \}$, from the same generation, are combined to build a coded packet:

\begin{equation}
	p' = \sum_{i \in W}c_i \times p_i
\end{equation}

The use of highly sparse coded packets (low $w$) would increase the throughput of the decoding operations~\cite{Feizi2014}. On the other hand, it can also lead to a greater probability of transmitting linear dependent combinations, thus increasing the corresponding network overhead and jeopardizing the performance. In this sense, based on the observation that the probability of generating linear dependent packets is higher as the transmissions evolves, \ac{TSNC} is proposed to tune the density throughout the transmission. The works~\cite{Feizi2014,Sorensen2015} show the high reduction on computational complexity while the overhead can be slightly increase.

Trullols~\emph{et al.}~\cite{Trullols2011} derived, for the \ac{RLNC} scheme, the exact decoding probability for a successful decoding event after the reception of $N$ coded packets. Zhao \emph{et al.} proposed in~\cite{Zhao2012} a slightly simplified modification, proposing a generalization of the original expression. However, the existing model that analytically characterize sparse coding scheme are far from being accurate or have some limitations, for instance the approximations proposed in~\cite{Heide2013,Feizi2014, Sorensen2014}. In particular, the authors of~\cite{Sorensen2015, Feizi2014, Tassi2016} exploit a lower bound to find a trade-off between complexity and overhead; such lower bound establishes that the probability of receiving a new linearly independent packet by the decoder, when it already has $i$ linearly independent packets is: 

\begin{equation}
P(i,k,d) \ge 1-(1-d)^{k-i}
\label{Eq: bound}
\end{equation}

\noindent where $k$ is the generation size and $d = \frac{w}{k}$ is the corresponding coding density.


Following the approach proposed in~\cite{Sorensen2015} the encoder decide the most appropriate density in order to send packet with the lowest density and the total packets transmitted is not higher than a defined \emph{budget}. The information that is available at the encoder is, in most cases, rather limited; in this case, the destination reports, with a dedicated control packet, the number of useful packets that have been already received. Hence, a more precise model for such probability can lead to the proposal of mechanisms to strongly improve the performance of \ac{TSNC}.

\section{Markov Chain Model}
\label{Sec: Implementation}

We define a set of states $(r,c)$, where $r$ is the current rank of the decoding matrix and $c$ is the number of non-zero columns. Clearly, $r \leq c$. Based on these states, we can establish a discrete Markov Chain, $\mathcal{S}_q(w,k)$. Figure~\ref{fig:markovchain} shows an illustrative example for such chain, with $w = 3$ and $k = 10$. As can be seen, $\mathcal{S}$ is an absorbing process, since the state $(k,k)$ is absorbing and will be eventually reached.

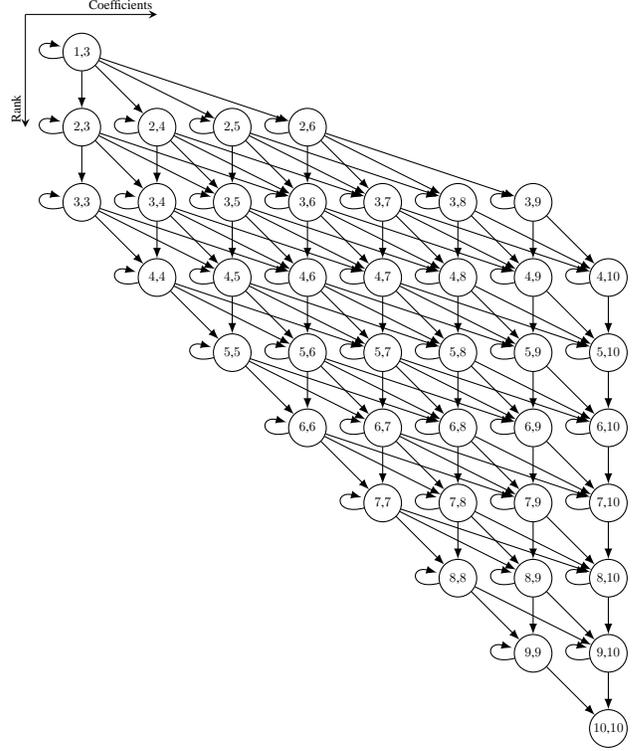
\begin{figure}[t]
\centering
\begin{tikzpicture}[scale=0.5,transform shape,state/.style={draw,circle,fill=white,inner sep=0pt,minimum width=1cm}]

\foreach \r in {2,...,10} {
	\pgfmathsetmacro{\aux}{round(3*\r)}
	\ifthenelse{\r < 3}{\pgfmathsetmacro{\m}{round(3)}}{\pgfmathsetmacro{\m}{\r}};
	\ifthenelse{10 < \aux}{\pgfmathsetmacro{\t}{round(10)}}{\pgfmathsetmacro{\t}{round(\aux)}};
	\foreach  \c in {\m,...,\t} {
		\pgfmathtruncatemacro\rn{int(\r)}
		\pgfmathtruncatemacro\cn{int(\c)}
 	\node[state] (estado_\rn\cn) at (2*\c,24-2*\r) {\pgfmathprintnumber[int trunc]{\r},\pgfmathprintnumber[int trunc]{\c}};
	\ifthenelse{10 = \r}{}{\path[-latex] (estado_\rn\cn) edge [loop left,->,>=latex] node {} (estado_\rn\cn);}
	}
}

\foreach \r in {2,...,10} {
	\pgfmathsetmacro{\aux}{int(3*\r)}
	\ifthenelse{\r < 3}{\pgfmathsetmacro{\m}{round(3)}}{\pgfmathsetmacro{\m}{\r}};
	\ifthenelse{10 < \aux}{\pgfmathsetmacro{\t}{round(10)}}{\pgfmathsetmacro{\t}{round(\aux)}};
	\foreach  \c in {\m,...,\t} {
	\ifthenelse{\r < \c}{
		\pgfmathtruncatemacro\rs{int(\r)}
		\pgfmathtruncatemacro\cs{int(\c)}
		\pgfmathtruncatemacro\re{int(1+\r)}
		\pgfmathtruncatemacro\ce{int({\c})}
				\draw[-latex] (estado_\rs\cs) -- (estado_\re\ce);}{}
		\foreach \j in {1,...,3} {
		\pgfmathtruncatemacro\rs{int(\r)}
		\pgfmathtruncatemacro\cs{int(\c)}
		\pgfmathtruncatemacro\re{int(1+\r)}
		\pgfmathtruncatemacro\ce{int({\c+\j})}
		\ifthenelse{\ce < 11}{
		\draw[-latex] (estado_\rs\cs) -- (estado_\re\ce);}{}
		}
	}
}

\node[state] (estado_13) at (6,22) {\pgfmathprintnumber[int trunc]{1},\pgfmathprintnumber[int trunc]{3}};
\foreach \j in {0,...,3} {
	\pgfmathtruncatemacro\ce{int({3+\j})}
	\draw[-latex] (estado_13) -- (estado_2\ce);

}
\path[-latex] (estado_13) edge [loop left,->,>=latex] node {} (estado_13);

\draw[-stealth] (4.5,23) -- node[pos=1,anchor=south east] {Coefficients} (8,23);
\draw[-stealth] (4.5,23) -- node[rotate=90,pos=1,anchor=south west] {Rank} (4.5,20);

\end{tikzpicture}
\caption{Markov chain for $w= 3$ and $k = 10$. We assume that $q > 1$, otherwise, the states $(2,3)$ and $(3,3)$ would not be feasible.}
\label{fig:markovchain}
\end{figure} 	

 \begin{theorem}(Transition Probabilities)
\label{th:transitionprobabilities}

The transition probability between states $(r,c)$ and $(r+i,c+j)$,  $p_{r,c}(i,j)$ are as follows:

\begin{subequations}
\renewcommand{\theequation}{\theparentequation.\arabic{equation}}
\begin{align}
 & \vartheta_q(r,c,w) \prod_{t=0}^{w-1} \frac{c-t}{k-t} & {} & \text{if} \begin{cases} i = 0\\ j = 0 \end{cases} \\
 & \left[ 1 - \vartheta_q(r,c,w) \right] \prod_{t=0}^{w-1} \frac{c-t}{k-t} & {} & \text{if} \begin{cases} i = 1\\ j = 0 \end{cases} \\
 & \binom{w}{i} \frac{\prod\limits_{t = 0}^{w-i-1} \left(c-t\right) \prod\limits_{t=c}^{c+i-1} \left(k - t \right)}{\prod\limits_{t = 0}^{w-1} \left(k-t\right)} & {} & \text{if} \begin{cases} i = 1\\ j = 1\ldots w  \end{cases} \\
 & \; \;  0 & {} &  \text{otherwise}
\end{align}
\end{subequations}

\noindent where $\vartheta_q(r,c,w)$ is the probability for a randomly generated sparse vector, with $w$ non-zero elements, to be linearly dependent with the already received $r$ vectors, from the $c$-dimensional space, assuming a \emph{Galois Field} $GF(2^q)$.

\end{theorem}

\begin{proof}
	We use simple combinatorial mathematics to derive the transition probabilities of the absorbing Markov chain that was previously discussed.
	
	For any $(r,c)$ state, if the new packet includes any novel coefficient, the rank is always increased. The number of combinations that would change the state from $(r,c)$ to $(r+1,c+i), \quad i>0$ are:
	
	\begin{equation}
	N_{(r,c)}(r+1,c+i) = \mathbb{C}_{w-i}^c \; \mathbb{C}_{i}^{k-c} =  \binom{c}{w-i} \cdot \binom{k-c}{i}
	\end{equation}
	
	\noindent where $\mathbb{C}_{t}^n$ is the combination of $n$ elements taken $t$ at a time without repetition. 
	
	In addition, the overall number of possible vectors are $\mathbb{C}_{w}^k$. Hence, the corresponding probability is:
	
	\begin{multline}
	p_{r,c}(1,j>0) = \frac{\binom{c}{w-i} \cdot \binom{k-c}{i}}{\binom{k}{w}} = \\ = \frac{\frac{c!}{\left(c-w+i\right)!\left(w-i\right)!} \cdot \frac{\left(k-c\right)!}{\left(k-c-i\right)!i!}}{\frac{k!}{\left(k-w\right)!w!}} = \\
	= \binom{w}{i} \cdot  \frac{\prod\limits_{t=0}^{w-i-1} \left(c-t \right) \cdot \prod\limits_{t=c}^{c+i-1} \left(k-t \right)}{\prod\limits_{t=0}^{w-1} \left(k-t \right)}
	\end{multline}
	
	On the other hand, if the new packet does not include any new coefficient, the corresponding vector could be either linearly dependent or independent, and this is established by $\vartheta$. The combinations that do not increase the number of already received coefficients is $\mathbb{C}_{w}^c$. Hence, the probability of staying at the current state, $(r,c)$ is:
	
	\begin{multline}
	p_{r,c}(0,0) = \vartheta_q(r,c,w) \frac{\binom{c}{w}}{\binom{k}{w}} = \vartheta_q(r,c,w) \frac{\frac{c!}{\left(c-w\right)! w!}}{\frac{k!}{\left(k-w\right)! w!}} = \\ = \vartheta_q(r,c,w) \prod\limits_{t=0}^{w-1} \frac{c-t}{k-t}
	\end{multline}
	
	While the probability of going to $(r+1,c)$ can be calculated as follows:
	
	\begin{multline}
	p_{r,c}(1,0) = \left[ 1- \vartheta_q(r,c,w)\right] \frac{\binom{c}{w}}{\binom{k}{w}} = \ldots = \\ = \left[1-\vartheta_q(r,c,w) \right] \prod\limits_{t=0}^{w-1} \frac{c-t}{k-t}
	\end{multline}
	
	Whenever a packet is received the rank can only increase in one single unit, and thus, $p_{r,c}(i,j) = 0, \quad i > 1$. Likewise, the number of novel coefficients per packet cannot be larger than $w$, so $p_{r,c}(i,j) = 0, \quad j > w$.
	
\end{proof}

\subsection{Empirical modeling of $\vartheta_q(r,c,w)$}

To the best of our knowledge there is not a closed expression for $\vartheta_q(r,c,w)$, so we conducted a Montecarlo analysis to empirically obtain it. We start by synthetically reaching every state of the corresponding \emph{Markov} chain, and then we generated a new vector, enforcing that the $w$ components were only selected from the $c$ already received coefficients. Afterwards we calculated the rank of the corresponding matrix to see whether it had increased; if that was the case the vector was linearly independent. We estimated the corresponding probability by counting the number of successes over a total of $100000$ independent experiments. Note that $\vartheta$ is the probability for a generated vector to be linearly dependent and thus $1 - \vartheta$ corresponds to the probability for a generated vector to be linearly independent, in both cases knowing that the number of received coefficients remains the same, i.e. no novel coefficients are used.

\begin{figure*}[t]
	\centering
	\def \figurewidth {0.2\textwidth}
	\def \figureheight {0.16\textwidth}
	\subfloat[$q = 1, w = 3$]{\pgfplotstableread{phi_fitting_w3_q1.dat}{\dataset} 

\begin{tikzpicture} 
font = \scriptsize, 
\begin{axis}[scale only axis, 
    width=\figurewidth,
    height=\figureheight,
    mark options={solid},
    ymin=0,
    ymax=1,
	xmin=0.4,
	xmax=1,
    ylabel={$\vartheta(r,c,w)$},
    xlabel={$r/c$},
    compat=1.3,
    legend pos= north west,
	legend cell align=left,
	legend style={draw=none, fill=none,legend columns=1,font = \scriptsize}
    ]

\addplot [color=black, mark=square*, mark size=1.5pt, only marks]
table[x index = 0, y index =1] from \dataset;
\addlegendentry{$C, \gamma = \lbrace 10, 5.74 \rbrace$}

\addplot[draw=black,domain=0.4:1,samples=100, forget plot]{x^5.7402};

\addplot [color=black, mark=o, mark size=2pt, only marks, thick]
table[x index = 3, y index =4] from \dataset;
\addlegendentry{$C, \gamma = \lbrace 20, 11.30 \rbrace$}

\addplot[draw=black,domain=0.4:1,samples=100, forget plot]{x^11.2974};

\addplot [color=black, mark=triangle*, mark size=2pt, only marks, mark repeat=3]
table[x index = 6, y index =7] from \dataset;
\addlegendentry{$C, \gamma = \lbrace 63, 25.43 \rbrace$}

\addplot[draw=black,domain=0.4:1,samples=100, forget plot	]{x^25.4308};

\end{axis}
\end{tikzpicture}}
	\subfloat[$q = 1, w = 7$]{\pgfplotstableread{phi_fitting_w7_q1.dat}{\dataset} 

\begin{tikzpicture} 
font = \scriptsize, 
\begin{axis}[scale only axis, 
    width=\figurewidth,
    height=\figureheight,
    mark options={solid},
    ymin=0,
    ymax=1,
	xmin=0.6,
	xmax=1,
    yticklabels={},
    xlabel={$r/c$},
    compat=1.3,
	    legend pos= north west,
	legend cell align=left,
	legend style={draw=none, fill=none,legend columns=1,font = \scriptsize}
    ]

\addplot [color=black, mark=square*, mark size=1.5pt, only marks]
table[x index = 0, y index =1] from \dataset;
\addlegendentry{$C, \gamma = \lbrace 15, 9.64 \rbrace$}

\addplot[draw=black,domain=0.6:1,samples=100, forget plot]{x^9.6426};

\addplot [color=black, mark=o, mark size=2pt, only marks, thick]
table[x index = 3, y index =4] from \dataset;
\addlegendentry{$C, \gamma = \lbrace 25, 16.55 \rbrace$}

\addplot[draw=black,domain=0.6:1,samples=100, forget plot]{x^16.5477};

\addplot [color=black, mark=triangle*, mark size=2pt, only marks, mark repeat=3]
table[x index = 6, y index =7] from \dataset;
\addlegendentry{$C, \gamma = \lbrace 63, 43.02 \rbrace$}

\addplot[draw=black,domain=0.6:1,samples=100, forget plot]{x^43.0219};

\end{axis}
\end{tikzpicture}}
	\subfloat[$q = 1, w = 15$]{\pgfplotstableread{phi_fitting_w15_q1.dat}{\dataset} 

\begin{tikzpicture} 
font = \scriptsize, 
\begin{axis}[scale only axis, 
    width=\figurewidth,
    height=\figureheight,
    mark options={solid},
    ymin=0,
    ymax=1,
	xmin=0.8,
	xmax=1,
    yticklabels={},
    xlabel={$r/c$},
    compat=1.3,
	    legend pos= north west,
 	legend cell align=left,
	legend style={draw=none, fill=none,legend columns=1,font = \scriptsize}
   ]


\addplot  [color=black, mark=square*, mark size=1.5pt, only marks]
table[x index = 0, y index =2] from \dataset;
\addlegendentry{$C, \gamma = \lbrace 25, 16.48 \rbrace$}

\addplot[draw=black,domain=0.8:1,samples=100, forget plot]{x^16.4784};


\addplot [color=black, mark=o, mark size=2pt, only marks, thick]
table[x index = 3, y index =5] from \dataset;
\addlegendentry{$C, \gamma = \lbrace 35, 23.44 \rbrace$}

\addplot[draw=black,domain=0.8:1,samples=100, forget plot]{x^23.4445};


\addplot [color=black, mark=triangle*, mark size=2pt, only marks, mark repeat=1]
table[x index = 6, y index =8] from \dataset;
\addlegendentry{$C, \gamma = \lbrace 63, 42.89 \rbrace$}

\addplot[draw=black,domain=0.8:1,samples=100, forget plot]{x^42.8917};

\end{axis}
\end{tikzpicture}}
	\subfloat[$q = 1, w = 31$]{\pgfplotstableread{phi_fitting_w31_q1.dat}{\dataset} 

\begin{tikzpicture} 
font = \scriptsize, 
\begin{axis}[scale only axis, 
    width=\figurewidth,
    height=\figureheight,
    mark options={solid},
    ymin=0,
    ymax=1,
	xmin=0.9,
	xmax=1,
    xlabel={$r/c$},
    yticklabels={},
    compat=1.3,
	    legend pos= north west,
	legend cell align=left,
	legend style={draw=none, fill=none,legend columns=1,font = \scriptsize}
    ]


\addplot  [color=black, mark=square*, mark size=1.5pt, only marks]
table[x index = 0, y index =2] from \dataset;
\addlegendentry{$C, \gamma = \lbrace 45 30.38 \rbrace$}

\addplot[draw=black,domain=0.9:1,samples=100, forget plot]{x^30.3778};


\addplot [color=black, mark=o, mark size=2pt, only marks, thick]
table[x index = 3, y index =5] from \dataset;
\addlegendentry{$C, \gamma = \lbrace 55, 37.21 \rbrace$}

\addplot[draw=black,domain=0.9:1,samples=100, forget plot]{x^37.2104};
   

\addplot [color=black, mark=triangle*, mark size=2pt, only marks, mark repeat=1]
table[x index = 6, y index =8] from \dataset;
\addlegendentry{$C, \gamma = \lbrace 63, 42.95 \rbrace$}

\addplot[draw=black,domain=0.9:1,samples=100, forget plot]{x^42.9486};

\end{axis}
\end{tikzpicture}}\\
	\subfloat[$q = 3, w = 3$]{\pgfplotstableread{phi_fitting_w3_q3.dat}{\dataset} 

\begin{tikzpicture} 
font = \scriptsize, 
\begin{axis}[scale only axis, 
    width=\figurewidth,
    height=\figureheight,
    mark options={solid},
    ymin=0,
    ymax=1,
	xmin=0.8,
	xmax=1,
    ylabel={$\vartheta(r,c,w)$},
    xlabel={$r/c$},
    compat=1.3,
	    legend pos= north west,
	legend cell align=left,
	legend style={draw=none, fill=none,legend columns=1,font = \scriptsize}
    ]


\addplot [color=black, mark=square*, mark size=1.5pt, only marks]
table[x index = 0, y index =2] from \dataset;
\addlegendentry{$C, \gamma = \lbrace 10, 17.30 \rbrace$}

\addplot[draw=black,domain=0.8:1,samples=100, forget plot]{x^17.2977};


\addplot [color=black, mark=o, mark size=2pt, only marks, thick]
table[x index = 3, y index =5] from \dataset;
\addlegendentry{$C, \gamma = \lbrace 20, 25.14$}

\addplot[draw=black,domain=0.8:1,samples=100, forget plot]{x^25.1355};


\addplot [color=black, mark=triangle*, mark size=2pt, only marks, mark repeat=1]
table[x index = 6, y index =8] from \dataset;
\addlegendentry{$C, \gamma = \lbrace 63, 39.13$}

\addplot[draw=black,domain=0.8:1,samples=100, forget plot]{x^39.1256};

\end{axis}
\end{tikzpicture}}
	\subfloat[$q = 3, w = 7$]{\pgfplotstableread{phi_fitting_w7_q3.dat}{\dataset} 

\begin{tikzpicture} 
font = \scriptsize, 
\begin{axis}[scale only axis, 
    width=\figurewidth,
    height=\figureheight,
    mark options={solid},
    ymin=0,
    ymax=1,
	xmin=0.8,
	xmax=1,
    yticklabels={},
    xlabel={$r/c$},
    compat=1.3,
	    legend pos= north west,
	legend cell align=left,
	legend style={draw=none, fill=none,legend columns=1,font = \scriptsize}
    ]


\addplot[color=black, mark=square*, mark size=1.5pt, only marks]
table[x index = 0, y index =2] from \dataset;
\addlegendentry{$C, \gamma = \lbrace 15, 30.06 \rbrace$}

\addplot[draw=black,domain=0.8:1,samples=100, forget plot]{x^30.0566};


\addplot [color=black, mark=o, mark size=2pt, only marks, thick]
table[x index = 3, y index =5] from \dataset;
\addlegendentry{$C, \gamma = \lbrace 25, 50.63 \rbrace$}

\addplot[draw=black,domain=0.8:1,samples=100, forget plot]{x^50.6334};


\addplot [color=black, mark=triangle*, mark size=2pt, only marks, mark repeat=1]
table[x index = 6, y index =8] from \dataset;
\addlegendentry{$C, \gamma = \lbrace 63, 129.65 \rbrace$}
\addplot[draw=black,domain=0.8:1,samples=100, forget plot]{x^129.6488};

\end{axis}
\end{tikzpicture}}
	\subfloat[$q = 3, w = 15$]{\pgfplotstableread{phi_fitting_w15_q3.dat}{\dataset} 

\begin{tikzpicture} 
font = \scriptsize, 
\begin{axis}[scale only axis, 
    width=\figurewidth,
    height=\figureheight,
    mark options={solid},
    ymin=0,
    ymax=1,
	xmin=0.9,
	xmax=1,
    yticklabels={},
    xlabel={$r/c$},
    compat=1.3,
	    legend pos= north west,
	legend cell align=left,
	legend style={draw=none, fill=none,legend columns=1,font = \scriptsize}
    ]


\addplot[color=black, mark=square*, mark size=1.5pt, only marks]
table[x index = 0, y index =2] from \dataset;
\addlegendentry{$C, \gamma = \lbrace 25, 51.08 \rbrace$}

\addplot[draw=black,domain=0.9:1,samples=100, forget plot]{x^51.0845};


\addplot [color=black, mark=o, mark size=2pt, only marks, thick]
table[x index = 3, y index =5] from \dataset;
\addlegendentry{$C, \gamma = \lbrace 35, 71.23 \rbrace$}

\addplot[draw=black,domain=0.9:1,samples=100, forget plot]{x^71.2263};


\addplot [color=black, mark=triangle*, mark size=2pt, only marks, mark repeat=1]
table[x index = 6, y index =8] from \dataset;
\addlegendentry{$C, \gamma = \lbrace 63, 129.96 \rbrace$}

\addplot[draw=black,domain=0.9:1,samples=100, forget plot]{x^129.9576};

\end{axis}
\end{tikzpicture}}
	\subfloat[$q = 3, w = 31$]{\pgfplotstableread{phi_fitting_w31_q3.dat}{\dataset} 

\begin{tikzpicture} 
font = \scriptsize, 
\begin{axis}[scale only axis, 
    width=\figurewidth,
    height=\figureheight,
    mark options={solid},
    ymin=0,
    ymax=1,
	xmin=0.9,
	xmax=1,
    yticklabels={},
    xlabel={$r/c$},
    compat=1.3,
	    legend pos= north west,
	legend cell align=left,
	legend style={draw=none, fill=none,legend columns=1,font = \scriptsize}
    ]


\addplot [color=black, mark=square*, mark size=1.5pt, only marks]
table[x index = 0, y index =2] from \dataset;
\addlegendentry{$C, \gamma = \lbrace 45, 93.13 \rbrace$}
\addplot[draw=black,domain=0.9:1,samples=100, forget plot]{x^93.1272};


\addplot [color=black, mark=o, mark size=2pt, only marks, thick]
table[x index = 3, y index =5] from \dataset;
\addlegendentry{$C, \gamma = \lbrace 55, 112.70 \rbrace$}
\addplot[draw=black,domain=0.9:1,samples=100, forget plot]{x^112.6971};


\addplot [color=black, mark=triangle*, mark size=2pt, only marks, mark repeat=1]
table[x index = 6, y index =8] from \dataset;
\addlegendentry{$C, \gamma = \lbrace 63, 130.24 \rbrace$}
\addplot[draw=black,domain=0.9:1,samples=100, forget plot]{x^130.2450};

\end{axis}
\end{tikzpicture}}
	\caption{$\vartheta(r,c)$ Vs. $\frac{r}{c}$. Markers correspond to the values obtained with the Montecarlo analysis, while the lines are the fitting curves $\left(\frac{r}{c}\right)^\gamma$.}
	\label{Fig: phiFitting}
\end{figure*}
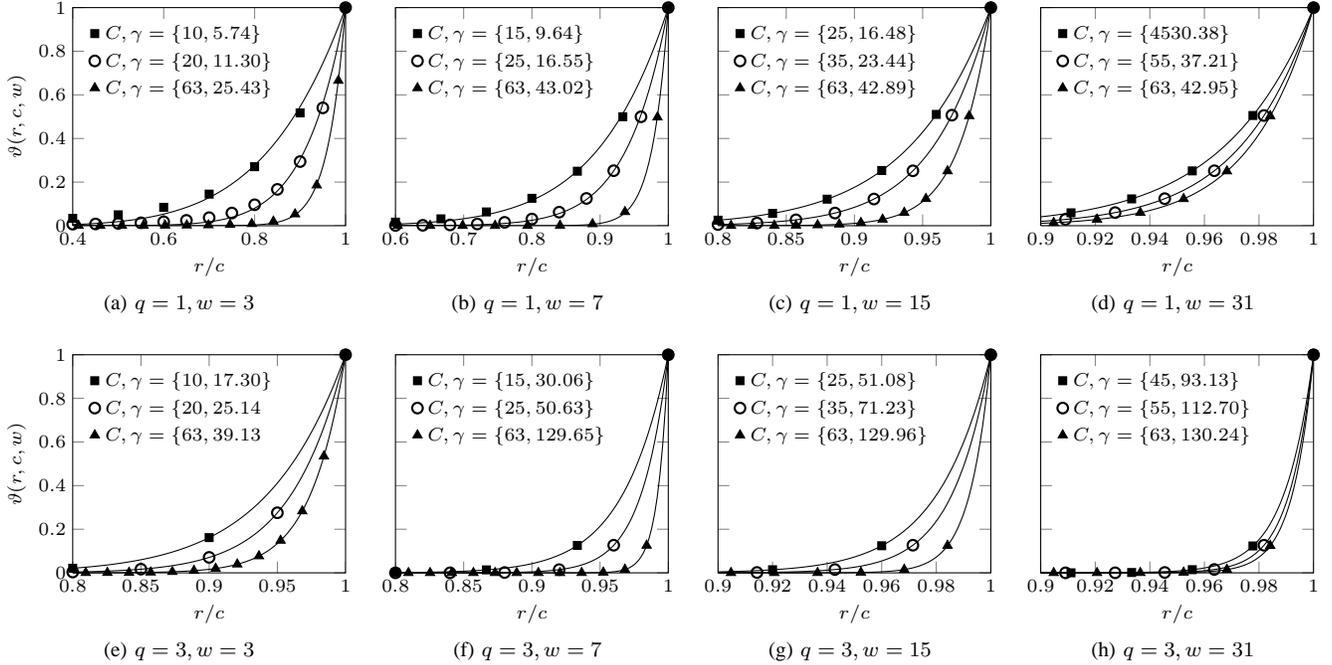

Figure~\ref{Fig: phiFitting} shows how $\vartheta$ varies against the ratio $\frac{r}{c}$ for various combinations of $q$, $w$ and $c$. In all cases we can use a function $\widetilde{\vartheta} = \left(\frac{r}{c}\right)^\gamma$ to approximate the observed behavior. Figure~\ref{Fig: phiFitting} also shows such fitting functions with a solid line, yielding a rather accurate approximation. This fitting is valid for all the possible combinations of $w \le frac{k/2}$ and $q$, and for every $c$ value (from 1 to $k$).

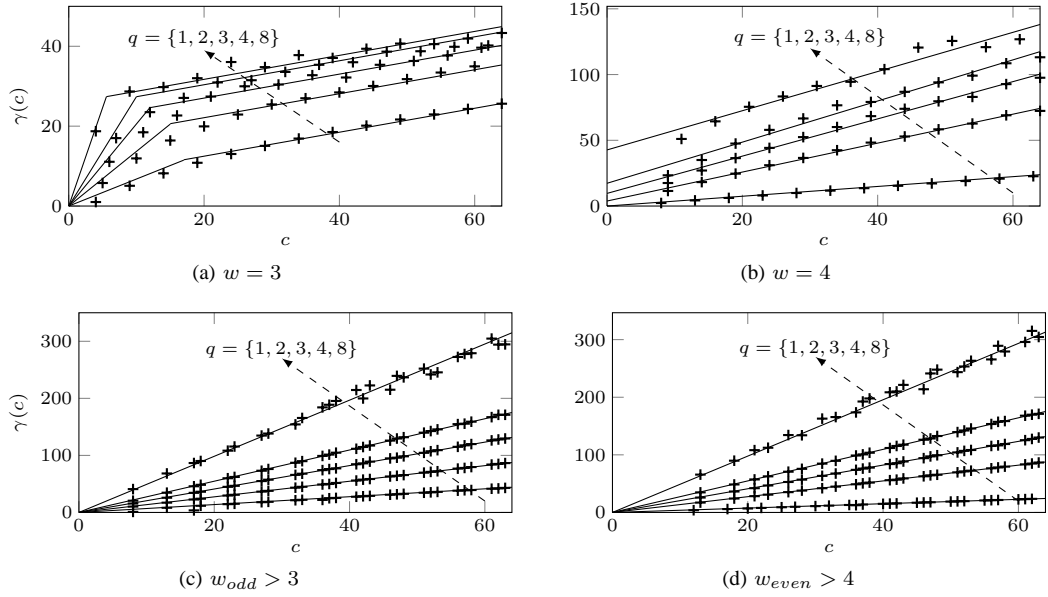
\begin{figure*}[t]
	\def \figurewidth {0.65\columnwidth}
	\def \figureheight {0.3\columnwidth}
	\begin{center}
	\subfloat[$w = 3$]{\pgfplotstableread{gamma_fitting_q1_w3.dat}{\dataset}
\pgfplotstableread{gamma_fitting_q2_w3.dat}{\datasetdos}   
\pgfplotstableread{gamma_fitting_q3_w3.dat}{\datasettres}   
\pgfplotstableread{gamma_fitting_q4_w3.dat}{\datasetcuatro}   
\pgfplotstableread{gamma_fitting_q8_w3.dat}{\datasetocho} 
\begin{tikzpicture} 
font = \scriptsize, 
\begin{axis}[scale only axis, 
    width=\figurewidth,
    height=\figureheight,
    mark options={solid},
    ymin=0,
    ymax=50,
	xmin=0,
	xmax=64,
    ylabel={$\gamma(c)$},
    xlabel={$c$},
    compat=1.3,
    legend pos= south east,
	legend style={draw=none, fill=none,legend columns=2, font = \footnotesize}
    ]

\addplot [color=black, mark=+, thick, only marks, mark repeat=5, mark phase = 1]
table[x index = 0, y index =1] from \dataset;


\addplot[black,domain=0:17, solid, forget plot] (x,0.6760*x);
\addplot[black,domain=17:64, solid, forget plot] (x,0.3*x+6.48);


\addplot [color=black, mark=+, thick, only marks, mark repeat=5, mark phase = 2]
table[x index = 0, y index =1] from \datasetdos;


\addplot[black,domain=0:15, solid, forget plot] (x,1.3672*x);
\addplot[black,domain=15:64, solid, forget plot] (x,0.3*x+16.12);


\addplot [color=black,mark=+, thick, only marks, mark repeat=5, mark phase = 3]
table[x index = 0, y index =1] from \datasettres;

\addplot[black,domain=0:12,  forget plot] (x,2.054*x);
\addplot[black,domain=12:64, forget plot] (x,0.3*x+21.03);


\addplot [color=black, mark=+, thick, only marks, mark repeat=5, mark phase = 4]
table[x index = 0, y index =1] from \datasetcuatro;

\addplot[black,domain=0:10, forget plot] (x,2.737*x);
\addplot[black,domain=10:64, forget plot] (x,0.3*x+24.37);


\addplot [color=black, mark=+, thick, only marks, mark repeat=5, mark phase = 0]
table[x index = 0, y index =1] from \datasetocho;

\draw (axis cs: 20,42) node[rotate=0] {$ q = \{1, 2, 3, 4, 8\}$};

\addplot[black,domain=5.56:64, forget plot] (x,0.3*x + 25.72);
\addplot[black,domain=0:5.56, forget plot] (x,4.92*x);

\draw[-latex,dashed] (axis cs: 40,16) -- (axis cs: 20,39);

\end{axis}
\end{tikzpicture}} 
	\subfloat[$w = 4$]{\pgfplotstableread{gamma_fitting_q1_w_even.dat}{\dataset}
\pgfplotstableread{gamma_fitting_q2_w_even.dat}{\datasetdos}   
\pgfplotstableread{gamma_fitting_q3_w_even.dat}{\datasettres}   
\pgfplotstableread{gamma_fitting_q4_w_even.dat}{\datasetcuatro}   
\pgfplotstableread{gamma_fitting_q8_w_even.dat}{\datasetocho}   

\begin{tikzpicture} 
font = \scriptsize, 
\begin{axis}[scale only axis, 
    width=\figurewidth,
    height=\figureheight,
    mark options={solid},
    ymin=0,
	xmin=0,
	xmax=64,
    xlabel={$c$},
    compat=1.3,
    legend pos= north west,
	legend style={draw=none, fill=none,legend columns=1, font = \footnotesize}
    ]

\addplot [color=black, mark=+, thick, only marks, mark repeat=5]
plot[]
table[x index = 2, y index =3] from \dataset;

\addplot[black,domain=0:64, solid, forget plot] (x,0.3724*x);


\addplot [color=black, mark=+, thick, only marks, mark repeat=5]
plot[]
table[x index = 2, y index =3] from \datasetdos;

\addplot[black,domain=0:64, solid, forget plot] (x,1.1010*x + 3.8170);


\addplot [color=black, mark=+, thick, only marks, mark repeat=5]
plot[]
table[x index = 2, y index =3] from \datasettres;

\addplot[black,domain=0:64, solid, forget plot] (x,1.4167*x+9.6235);


\addplot [color=black, mark=+, thick, only marks, mark repeat=5]
plot[]
table[x index = 2, y index =3] from \datasetcuatro;

\addplot[black,domain=0:64, solid, forget plot] (x,1.5646*x+17.3039);


\addplot [color=black, mark=+, thick, only marks, mark repeat=5, mark phase=3]
plot[]
table[x index = 2, y index =3] from \datasetocho;

\addplot[black,domain=0:64, solid, forget plot] (x,1.4908*x+42.6730);



\draw (axis cs: 30,128) node[rotate=0] {$q = \{1,2,3,4,8\}$};
\draw[-latex,dashed] (axis cs: 60,10) -- (axis cs: 30,120);

\end{axis}
\end{tikzpicture}}   \\
	\subfloat[$w_{odd} > 3$]{\pgfplotstableread{gamma_fitting_q1.dat}{\dataset}
\pgfplotstableread{gamma_fitting_q2.dat}{\datasetdos}   
\pgfplotstableread{gamma_fitting_q3.dat}{\datasettres}   
\pgfplotstableread{gamma_fitting_q4.dat}{\datasetcuatro}   
\pgfplotstableread{gamma_fitting_q8.dat}{\datasetocho} 
\begin{tikzpicture} 
font = \scriptsize, 
\begin{axis}[scale only axis, 
    width=\figurewidth,
    height=\figureheight,
    mark options={solid},
    ymin=0,
    ymax=350,
	xmin=0,
	xmax=64,
    ylabel={$\gamma(c)$},
    xlabel={$c$},
    compat=1.3,
    legend pos= north west,
	legend style={draw=none, fill=none,legend columns=1, font = \footnotesize}
    ]

\addplot [color=black, mark=+, thick, only marks, mark repeat=5]
plot[]
table[x index = 0, y index =1] from \dataset;

\addplot[black,domain=0:64, solid, forget plot] (x,0.69*x);

\addplot [color=black,  mark=+, thick, only marks, mark repeat=5,mark phase=2]
plot[]
table[x index = 2, y index =3] from \dataset;

\addplot [color=black,  mark=+, thick, only marks, mark repeat=5,,mark phase=4]
plot[]
table[x index = 4, y index =5] from \dataset;


\addplot [color=black, mark=+, thick, only marks, mark repeat=5]
plot[]
table[x index = 0, y index =1] from \datasetdos;

\addplot[black,domain=0:64, solid, forget plot] (x,1.3672*x);

\addplot [color=black, mark=+, thick, only marks, mark repeat=5,mark phase=2]
plot[]
table[x index = 2, y index =3] from \datasetdos;

\addplot [color=black, mark=+, thick, only marks, mark repeat=5,mark phase=4]
plot[]
table[x index = 4, y index =5] from \datasetdos;


\addplot [color=black, mark=+, thick, only marks, mark repeat=5]
plot[]
table[x index = 0, y index =1] from \datasettres;

\addplot[black,domain=0:64, solid, forget plot] (x,2.054*x);

\addplot [color=black, mark=+, thick, only marks, mark repeat=5,mark phase=2]
plot[]
table[x index = 2, y index =3] from \datasettres;

\addplot [color=black, mark=+, thick, only marks, mark repeat=5,mark phase=4]
plot[]
table[x index = 4, y index =5] from \datasettres;


\addplot [color=black, mark=+, thick, only marks, mark repeat=5]
plot[]
table[x index = 0, y index =1] from \datasetcuatro;

\addplot[black,domain=0:64, forget plot] (x,2.737*x);

\addplot [color=black, mark=+, thick, only marks, mark repeat=5,mark phase=2]
plot[]
table[x index = 2, y index =3] from \datasetcuatro;

\addplot[color=black, mark=+, thick, only marks, mark repeat=5,mark phase=4]
plot[]
table[x index = 4, y index =5] from \datasetcuatro;


\addplot[color=black, mark=+, thick, only marks, mark repeat=5]
plot[]
table[x index = 0, y index =1] from \datasetocho;

\addplot[black,domain=0:64, samples=10] (x,4.92*x);

\addplot [color=black, mark=+, thick, only marks, mark repeat=5,mark phase=2]
plot[]
table[x index = 2, y index =3] from \datasetocho;

\addplot [color=black, mark=+, thick, only marks, mark repeat=5,mark phase=4]
plot[]
table[x index = 4, y index =5] from \datasetocho;


\draw (axis cs: 30,285) node[rotate=0] {$q = \{1,2,3,4,8\}$};
\draw[-latex,dashed] (axis cs: 60,20) -- (axis cs: 30,270);

\end{axis}
\end{tikzpicture}} 
	\subfloat[$w_{even} > 4$]{\pgfplotstableread{gamma_fitting_q1_w_even.dat}{\dataset}
\pgfplotstableread{gamma_fitting_q2_w_even.dat}{\datasetdos}   
\pgfplotstableread{gamma_fitting_q3_w_even.dat}{\datasettres}   
\pgfplotstableread{gamma_fitting_q4_w_even.dat}{\datasetcuatro}   
\pgfplotstableread{gamma_fitting_q8_w_even.dat}{\datasetocho}   
\begin{tikzpicture} 
font = \scriptsize, 
\begin{axis}[scale only axis, 
    width=\figurewidth,
    height=\figureheight,
    mark options={solid},
    ymin=0,
	xmin=0,
	xmax=64,
    xlabel={$c$},
    compat=1.3,
    legend pos= north west,
	legend style={draw=none, fill=none,legend columns=1, font = \footnotesize}
    ]

\addplot [color=black, mark=+, thick, only marks, mark repeat=5]
plot[]
table[x index = 4, y index =5] from \dataset;

\addplot [color=black, mark=+, thick, only marks, mark repeat=5]
plot[]
table[x index = 6, y index =7] from \dataset;

\addplot [color=black, mark=+, thick, only marks, mark repeat=5]
plot[]
table[x index = 8, y index =9] from \dataset;

\addplot[black,domain=0:64, solid, forget plot] (x,0.3791*x);


\addplot [color=black, mark=+, thick, only marks, mark repeat=5]
plot[]
table[x index = 4, y index =5] from \datasetdos;

\addplot [color=black, mark=+, thick, only marks, mark repeat=5]
plot[]
table[x index = 6, y index =7] from \datasetdos;

\addplot [color=black, mark=+, thick, only marks, mark repeat=5]
plot[]
table[x index = 8, y index =9] from \datasetdos;

\addplot[black,domain=0:64, solid, forget plot] (x,1.3682*x);


\addplot [color=black, mark=+, thick, only marks, mark repeat=5]
plot[]
table[x index = 4, y index =5] from \datasettres;

\addplot [color=black, mark=+, thick, only marks, mark repeat=5]
plot[]
table[x index = 6, y index =7] from \datasettres;

\addplot [color=black, mark=+, thick, only marks, mark repeat=5]
plot[]
table[x index = 8, y index =9] from \datasettres;

\addplot[black,domain=0:64, solid, forget plot] (x,2.0561*x);


\addplot [color=black, mark=+, thick, only marks, mark repeat=5]
plot[]
table[x index = 4, y index =5] from \datasetcuatro;

\addplot [color=black, mark=+, thick, only marks, mark repeat=5]
plot[]
table[x index = 6, y index =7] from \datasetcuatro;

\addplot [color=black, mark=+, thick, only marks, mark repeat=5]
plot[]
table[x index = 8, y index =9] from \datasetcuatro;

\addplot[black,domain=0:64, solid, forget plot] (x,2.7380*x);


\addplot [color=black, mark=+, thick, only marks, mark repeat=5]
plot[]
table[x index = 4, y index =5] from \datasetocho;

\addplot [color=black, mark=+, thick, only marks, mark repeat=5]
plot[]
table[x index = 6, y index =7] from \datasetocho;

\addplot [color=black, mark=+, thick, only marks, mark repeat=5]
plot[]
table[x index = 8, y index =9] from \datasetocho;

\addplot[black,domain=0:64, solid, forget plot] (x,4.8913*x);



\draw (axis cs: 30,285) node[rotate=0] {$q = \{1,2,3,4,8\}$};
\draw[-latex,dashed] (axis cs: 60,20) -- (axis cs: 30,270);

\end{axis}
\end{tikzpicture}}
	\caption{Fitting of $\gamma_q(c)$ function. Markers correspond to the values obtained with the Montecarlo analysis, while the lines are the fitting curves}
	\label{Fig: GammaFitting}
	\end{center}
\end{figure*}



Figure~\ref{Fig: GammaFitting} shows the evolution of $\gamma$ against $c$ for different $w$ and $q$ configurations. As can be seen on the lower figures, for $w>3$, there is a clear linear relationship between $\gamma$ and $c$, and the slope of the corresponding line does only depends on a even or odd value of $w$ and on $q$. On the other hand, for $w=3$ and $w=4$ a different behaviour was observed (see upper figures). In order not to increase the complexity of the model, we have approximated the behaviour of $w=3$ with two different lines, since it can be seen that, for lager $c$'s, the slope of the corresponding function does not depend on $q$, being all of them parallel. As will be seen later, the results are rather accurate, despite there is a non-negligible difference between the observed values and the corresponding fitting. 

With all of the above into account, we can use the following function to estimate the value of $\gamma$:

\begin{align}
& \begin{cases} \gamma = m_{\text{odd}} \cdot c & c < c_0 \\ \gamma = 0.3 \left [  c -  c_0 \left(1 - m_{\text{odd}}\right) \right]& c \geq c_0 \end{cases} & & w &= 3  \\
& \quad \gamma = m_{w_4} \cdot c + b_4    & & w &= 4 \\
& \quad \gamma = m_{\text{even}/\text{odd}} \cdot c     & 4 < & w & < k/2
\end{align}

\noindent where the value of the slope $m$ depends on $q$ and $c_0$ is the point where the slope of $\gamma$ changes for $w = 3$. The corresponding values of $m$ and $c_0$ for the different $q$'s are given in Table~\ref{tab:m_co_q}

\begin{table}
\begin{center}
\caption{$m$ and $c_0$ for different $q$ values}
\label{tab:m_co_q}
\begin{tabular}{|c|c|c|c|c|c|}
\hline
$q$ & 1 & 2 & 3 & 4 & 8 \\
\hline
$m_{\text{odd}}$ & 0.676 & \multirow{2}{*}{1.367} & \multirow{2}{*}{2.055} & \multirow{2}{*}{2.738} & \multirow{2}{*}{4.891} \\

$m_{\text{even}}$ & 0.337 &  &  &  &  \\
\hline
$c_0$ & 17 & 15 & 12 & 10 & 6  \\
\hline
$m_{w_4}$ & 0.337 & 1.101 & 1.417 & 1.565 & 1.491  \\
$b_{w_4}$ & 0 & 3.817 & 9.627 & 17.298 & 42.634  \\
\hline
\end{tabular}
\end{center}
\end{table}

\subsection{Fundamental Matrix}

Once we have established all the transition probabilities, we can build the fundamental matrix for the absorbing Markov chain. According to~\cite{kemeny1960}, the canonical form, for $t$ and $r$ transient and absorbing states, respectively, of such matrix can be defined as follows:

\begin{equation}
P = \left[ \begin{matrix}
I_{r \times r} & 0 \\
R_{t \times r} & Q_{t \times t}
\end{matrix} \right]
\end{equation}

Since there is only one absorbing state $(k,k)$, $I$ is an identity matrix of one single element and $R$ is a column vector with the transition probabilities of all the remaining states to $(k,k)$. Finally $Q$ is a matrix with the transition probabilities between the transient states. We will assume that the first row/column of this matrix correspond to the initial state, i.e. $(1,w)$.






\begin{theorem}{Average number of transitions (Theorem 3.2.4 in~\cite{kemeny1960}).}
	The average number of transitions before being absorbed when starting in a transient state i is the i-th element of the column vector
	
	\begin{equation}
	M = \left( I - Q \right)^{-1} \Gamma
	\label{Eq: Transmissions}	
	\end{equation}
	
\noindent where $I$ is an identity matrix with the same dimension as $Q$, and $\Gamma$ is an all-one column vector. Furthermore, the matrix $N = \left( I - Q \right)^{-1}$ is called the \emph{fundamental matrix} for $P$.	

\end{theorem}

\begin{corollary}{Average number of transmissions.}
	The average number of transitions defines the average number of transmissions, since in our case, a transition always correspond to a packet transmission, and the first element of $M$ would correspond to the number of transitions that are required to hit the $(k,k)$ state from the initial one.
\end{corollary}

\begin{theorem}{Probability of being in state $j$ after $N$ transitions (Theorem 3.1.1 in~\cite{kemeny1960}).}
	The entry $p_{ij}^{(N)}$ of $P^N$ is the probability of being in state $j$ after $N$ transitions, provided that the chain was started in $i$
	
\end{theorem}

\begin{corollary}{Probability of succesfully decode a generation.}
	Since the chain always starts from the initial state $1 \rightarrow (1,w)$ and there is only one absorbing state, $t+r \rightarrow (k,k)$, the probability of successfully decode a generation after $N$ transmissions is:
	$\xi(N) = p_{1,\textbf{t+r}}^{(N)}$
\end{corollary}







Another parameter of interest would be the probability of increasing the rank of the corresponding decoding matrix with every transmission. This would allow establishing dynamic tuning schemes for the coding density, since lower densities would yield lower coding/decoding times, but they might as well lead to a higher number of transmissions.

\begin{theorem}{Transient Probabilities (Theorem 3.5.7 in~\cite{kemeny1960}).}
	The probability of visiting a state $j$ when starting a transient state $i$ is the $(i,j)$ entry of the transient probabilities matrix $H$:
	
	\begin{equation}
	H = \left( N - I \right) \cdot N_d^{-1}
	\end{equation}
	
	\noindent where $N_d$ is a diagonal matrix with the same diagonal of $N$.
\end{theorem}

\begin{corollary}{Probability of receiving a linearly independent packet.}
	We define a set of states $\mathbf{s}(r)$ as all the states from the chain where the rank equals $r$, $\mathbf{s}(r) = \forall (r,i) \in \mathcal{S}$. Hence, the probability of increasing the rank of the matrix when $r$ independent packets have been already received can be calculated as follows:
	
	\begin{equation}
	\delta(r) = \sum\limits_{\forall i \in \mathbf{s}(r)} H(0,i) \cdot \left( 1 - p_{r,c}(0,0) \right)
	\label{Eq: ProbIncrRank}
	\end{equation} 
\end{corollary}

\subsection{Impact of errors}

So far we have assumed and ideal wireless channel between the transmitter and the receiver. wever the model can be easily broadened so as to consider packet-erasure links. For that we just need to modify the corresponding transition probabilities as follows:

\begin{equation}
\widetilde{p_{r,c}}(i,j) = \begin{cases} p_{r,c}(i,j) \left(1 - \alpha \right) & (i,j) \neq (0,0) \\ p_{r,c}(i,j) \left(1-\alpha \right) + \alpha \quad & (i,j) = (0,0) \end{cases}
\end{equation}

\noindent where $\alpha$ is the frame error rate of the wireless link.

We can, for instance see, that $\widetilde{Q} = \left(1 - \alpha \right) Q + \alpha I$, and hence:

\begin{multline}
\widetilde{M} = \left( I - \widetilde{Q} \right)^{-1} \Gamma = \left[I - \left(\left(1 - \alpha \right) Q + \alpha I \right) \right]^{-1} = \\ = \frac{\left( I - Q \right)^{-1}}{1-\alpha}
\end{multline}

\section{Simulation and Model Validation}
\label{Sec: Results}

In this Section we assess the validity of the proposed model and all the results that were previously discussed by means of an extensive simulation campaign. We use the M4RIE library~\cite{DBLP:journals} to transmit (using TSNC) 10000 different generations, in order to ensure statistical tightness of the corresponding results.

Figure~\ref{Fig: ProbabilityIncreasingRank} shows the probability of receiving a linearly independent packet against the current rank at the receiver, Eq. \ref{Eq: ProbIncrRank}. As can be seen, the probability is close to 1 until the rank equals a large value. Afterwards, it decreases quite sharply, especially for low values of $q$ and $w$. We can also see the difference between the simulation results (solid line) and the proposed model (markers). The model show a good accuracy, the  mean squared error in the worst case, $k=128$, $q=3$, is $3.14 \cdot 10^{-4}$. On the other hand, the lower bound that was previously discussed, which has been used in various works until now, yields a much lower probability, with very little accuracy. Since the corresponding operations that need to be performed are much faster if the coding density ($w$) is low, this result is really interesting, since it shows that there is not any disadvantage in using a low $w$ until the rank at the receiver is quite high; at that moment increasing $w$ would probably improve the corresponding performance.

\begin{figure*}[t]
	\centering
	\def \figurewidth {0.2\textwidth}
	\def \figureheight {0.16\textwidth}	
	\subfloat[][$q=1, k=64$]{  \pgfplotstableread{probability_increasingRank_q1.dat}{\dataset}
\begin{tikzpicture} 
font = \scriptsize, 
\begin{axis}[scale only axis, 
    width=\figurewidth,
    height=\figureheight,
    mark options={solid},
    ymin=0,
    ymax=1,
	xmin=1,
	xmax=64,
    ylabel={$\gamma(c)$},
    xlabel={$c$},
    compat=1.3,
    legend pos= south west,
		legend cell align=left,
	legend style={draw=none, fill=none,legend columns=1, font = \footnotesize}
    ]

\addplot [color=black, mark=square*, mark size=1.2pt,only marks]
plot[]
table[x index = 0, y index =1] from \dataset;
\addlegendentry{$w=3$}

\addplot [color=black, dashed, mark=square,mark size=1.5pt, mark repeat=5]
plot[]
table[x index = 0, y index =3] from \dataset;
\addlegendentry{Bound$_{w=3}$}

\addplot [color=black, mark=*, mark size=1.2pt,only marks]
plot[]
table[x index = 4, y index =5] from \dataset;
\addlegendentry{$w=15$}

\addplot [color=black, solid, forget plot]
plot[]
table[x index = 0, y index =2] from \dataset;

\addplot [color=black, solid,forget plot]
plot[]
table[x index = 4, y index =6] from \dataset;

\addplot [color=black, dashed, mark=o, mark size=1.5pt, mark repeat=5]
plot[]
table[x index = 4, y index =7] from \dataset;
\addlegendentry{Bound$_{w=15}$}

\end{axis}
\end{tikzpicture}}
	\subfloat[][$q=1, k=64, c \in (40, 64)$]{\pgfplotstableread{probability_increasingRank_q1.dat}{\dataset}
\begin{tikzpicture} 
font = \scriptsize, 
\begin{axis}[scale only axis, 
    width=\figurewidth,
    height=\figureheight,
    mark options={solid},
    ymin=0,
    ymax=1,
	xmin=50,
	xmax=63,
    yticklabels={},
    xlabel={$c$},
    compat=1.3,
    legend pos= south east,
	legend style={draw=none, fill=none,legend columns=1, font = \footnotesize}
    ]

\addplot [color=black, mark=square*, only marks, mark size=1.2pt]
plot[]
table[x index = 0, y index =1] from \dataset;

\addplot [color=black, mark=*, only marks, mark size=1.5pt]
plot[]
table[x index = 4, y index =5] from \dataset;

\addplot [color=black, solid, forget plot,  mark size=1.5pt]
plot[]
table[x index = 0, y index =2] from \dataset;

\addplot [color=black, dashed, mark=square, mark size=1.2pt, mark repeat=3]
plot[]
table[x index = 0, y index =3] from \dataset;

\addplot [color=black, solid,forget plot]
plot[]
table[x index = 4, y index =6] from \dataset;

\addplot [color=black, dashed, mark=o, mark size=1.5pt, mark repeat=3]
plot[]
table[x index = 4, y index =7] from \dataset;

\end{axis}
\end{tikzpicture}}	
	\subfloat[][$q=1, k=128$]{  \pgfplotstableread{probability_increasingRank_q1_128.dat}{\dataset}
\begin{tikzpicture} 
font = \scriptsize, 
\begin{axis}[scale only axis, 
    width=\figurewidth,
    height=\figureheight,
    mark options={solid},
    ymin=0,
    ymax=1,
	xmin=1,
	xmax=128,
    ylabel={},
    yticklabels={}, 
    xlabel={$c$},
    compat=1.3,
    legend pos= south west,
    legend cell align=left,
	legend style={draw=none, fill=none,legend columns=1, font = \footnotesize}
    ]

\addplot [color=black, mark=square*, mark size=1.2pt,only marks, mark repeat=2]
plot[]
table[x index = 0, y index =1] from \dataset;
\addlegendentry{$w=3$}

\addplot [color=black, dashed, mark=square,mark size=1.5pt, mark repeat=9]
plot[]
table[x index = 0, y index =3] from \dataset;
\addlegendentry{Bound$_{w=3}$}

\addplot [color=black, mark=*, mark size=1.2pt,only marks, mark repeat=9]
plot[]
table[x index = 4, y index =5] from \dataset;
\addlegendentry{$w=15$}

\addplot [color=black, solid, forget plot]
plot[]
table[x index = 0, y index =2] from \dataset;

\addplot [color=black, solid,forget plot]
plot[]
table[x index = 4, y index =6] from \dataset;

\addplot [color=black, dashed, mark=o, mark size=1.5pt, mark repeat=5]
plot[]
table[x index = 4, y index =7] from \dataset;
\addlegendentry{Bound$_{w=15}$}

\end{axis}
\end{tikzpicture}}
	\subfloat[][$q=1, k=128, c \in (100, 128)$]{\pgfplotstableread{probability_increasingRank_q1_128.dat}{\dataset}
\begin{tikzpicture} 
font = \scriptsize, 
\begin{axis}[scale only axis, 
    width=\figurewidth,
    height=\figureheight,
    mark options={solid},
    ymin=0,
    ymax=1,
	xmin=100,
	xmax=127,
    ylabel={},
    yticklabels={},
    xlabel={$c$},
    compat=1.3,
    legend pos= south west,
    legend cell align=left,
	legend style={draw=none, fill=none,legend columns=1, font = \footnotesize}
    ]

\addplot [color=black, mark=square*, only marks,  mark size=1.2pt]
plot[]
table[x index = 0, y index =1] from \dataset;

\addplot [color=black, mark=*, only marks,  mark size=1.5pt]
plot[]
table[x index = 4, y index =5] from \dataset;

\addplot [color=black, solid, forget plot]
plot[]
table[x index = 0, y index =2] from \dataset;

\addplot [color=black, dashed, mark=square, mark size=1.2pt, mark repeat=5]
plot[]
table[x index = 0, y index =3] from \dataset;

\addplot [color=black, solid,forget plot]
plot[]
table[x index = 4, y index =6] from \dataset;

\addplot [color=black, dashed, mark=o, mark size=1.5pt, mark repeat=5]
plot[]
table[x index = 4, y index =7] from \dataset;

\end{axis}
\end{tikzpicture}}  \\
	
	\subfloat[][$q=3, k=64$]{  \pgfplotstableread{probability_increasingRank_q3.dat}{\dataset}
\begin{tikzpicture} 
font = \scriptsize, 
\begin{axis}[scale only axis, 
width=\figurewidth,
height=\figureheight,
mark options={solid},
ymin=0,
ymax=1,
xmin=1,
xmax=64,
ylabel={$\gamma(c)$},
xlabel={$c$},
compat=1.3,
    legend pos= south west,
    legend cell align=left,
legend style={draw=none, fill=none,legend columns=1, font = \footnotesize}
]

\addplot [color=black, mark=square*, mark size=1.2pt,only marks]
plot[]
table[x index = 0, y index =1] from \dataset;
\addlegendentry{$w=3$}

\addplot [color=black, dashed, mark=square,mark size=1.5pt, mark repeat=5]
plot[]
table[x index = 0, y index =3] from \dataset;
\addlegendentry{Bound$_{w=3}$}

\addplot [color=black, mark=*, mark size=1.2pt,only marks]
plot[]
table[x index = 4, y index =5] from \dataset;
\addlegendentry{$w=15$}

\addplot [color=black, solid, forget plot]
plot[]
table[x index = 0, y index =2] from \dataset;

\addplot [color=black, solid,forget plot]
plot[]
table[x index = 4, y index =6] from \dataset;

\addplot [color=black, dashed, mark=o, mark size=1.5pt, mark repeat=5]
plot[]
table[x index = 4, y index =7] from \dataset;
\addlegendentry{Bound$_{w=15}$}

\end{axis}
\end{tikzpicture}}
	\subfloat[][$q=3, k=64, c \in (40, 64)$]{ \pgfplotstableread{probability_increasingRank_q3.dat}{\dataset}
\begin{tikzpicture} 
font = \scriptsize, 
\begin{axis}[scale only axis, 
    width=\figurewidth,
    height=\figureheight,
    mark options={solid},
    ymin=0,
    ymax=1,
	xmin=40,
	xmax=63,
    yticklabels={},
    xlabel={$c$},
    compat=1.3,
    legend pos= south east,
	legend style={draw=none, fill=none,legend columns=1, font = \footnotesize}
    ]

\addplot [color=black, mark=square*, only marks,  mark size=1.2pt]
plot[]
table[x index = 0, y index =1] from \dataset;

\addplot [color=black, mark=*, only marks,  mark size=1.2pt]
plot[]
table[x index = 4, y index =5] from \dataset;

\addplot [color=black, solid, forget plot]
plot[]
table[x index = 0, y index =2] from \dataset;

\addplot [color=black, dashed, mark=square, mark size=1.5pt, mark repeat=5]
plot[]
table[x index = 0, y index =3] from \dataset;

\addplot [color=black, solid,forget plot]
plot[]
table[x index = 4, y index =6] from \dataset;

\addplot [color=black, dashed, mark=o, mark size=1.5pt, mark repeat=5]
plot[]
table[x index = 4, y index =7] from \dataset;

\end{axis}
\end{tikzpicture}}
	\subfloat[][$q=3, k=128$]{  \pgfplotstableread{probability_increasingRank_q3_128.dat}{\dataset}
\begin{tikzpicture} 
font = \scriptsize, 
\begin{axis}[scale only axis, 
    width=\figurewidth,
    height=\figureheight,
    mark options={solid},
    ymin=0,
    ymax=1,
	xmin=1,
	xmax=128,
    ylabel={},
    yticklabels={}, 
    xlabel={$c$},
    compat=1.3,
    legend pos= south west,
    legend cell align=left,
	legend style={draw=none, fill=none,legend columns=1, font = \footnotesize}
    ]

\addplot [color=black, mark=square*, mark size=1.2pt,only marks]
plot[]
table[x index = 0, y index =1] from \dataset;
\addlegendentry{$w=3$}

\addplot [color=black, dashed, mark=square,mark size=1.5pt, mark repeat=9]
plot[]
table[x index = 0, y index =3] from \dataset;
\addlegendentry{Bound$_{w=3}$}

\addplot [color=black, mark=*, mark size=1.2pt,only marks]
plot[]
table[x index = 4, y index =5] from \dataset;
\addlegendentry{$w=15$}

\addplot [color=black, solid, forget plot]
plot[]
table[x index = 0, y index =2] from \dataset;

\addplot [color=black, solid,forget plot]
plot[]
table[x index = 4, y index =6] from \dataset;

\addplot [color=black, dashed, mark=o, mark size=1.5pt, mark repeat=5]
plot[]
table[x index = 4, y index =7] from \dataset;
\addlegendentry{Bound$_{w=15}$}

\end{axis}
\end{tikzpicture}}
	\subfloat[][$q=3, k=128, c \in (100, 128)$]{ \pgfplotstableread{probability_increasingRank_q3_128.dat}{\dataset}
\begin{tikzpicture} 
font = \scriptsize, 
\begin{axis}[scale only axis, 
    width=\figurewidth,
    height=\figureheight,
    mark options={solid},
    ymin=0,
    ymax=1,
	xmin=100,
	xmax=127,
    ylabel={},
    yticklabels={},
    xlabel={$c$},
    compat=1.3,
    legend pos= south east,
	legend style={draw=none, fill=none,legend columns=1, font = \footnotesize}
    ]

\addplot [color=black, mark=square*, only marks,  mark size=1.2pt]
plot[]
table[x index = 0, y index =1] from \dataset;

\addplot [color=black, mark=*, only marks,  mark size=1.5pt]
plot[]
table[x index = 4, y index =5] from \dataset;

\addplot [color=black, solid, forget plot]
plot[]
table[x index = 0, y index =2] from \dataset;

\addplot [color=black, dashed, mark=square, mark size=1.2pt, mark repeat=9]
plot[]
table[x index = 0, y index =3] from \dataset;

\addplot [color=black, solid,forget plot]
plot[]
table[x index = 4, y index =6] from \dataset;

\addplot [color=black, dashed, mark=o, mark size=1.5pt, mark repeat=9]
plot[]
table[x index = 4, y index =7] from \dataset;

\end{axis}
\end{tikzpicture}}
	\caption{Probability of increasing rank Vs. current rank at the decoder for different field sizes ($q=1,3$) and generation sizes ($k=64,128$). Markers show the values obtained with the proposed model and the line corresponds to the simulation results}
	\label{Fig: ProbabilityIncreasingRank}
\end{figure*}
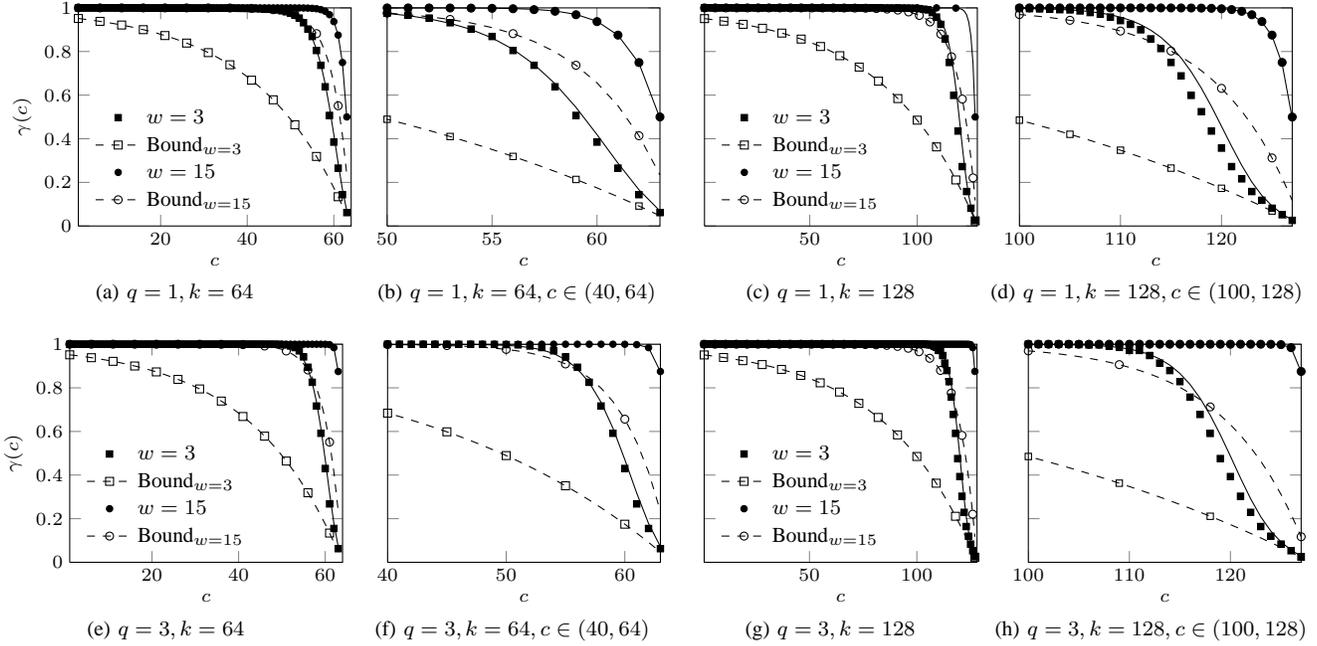

Another interesting result would be the average number of transmissions that would be required to correctly receive a generation, Eq.~\ref{Eq: Transmissions}. Table~\ref{Tab: Transmissions} collects such metric for different configurations, in which we have modified the value of both $w$ and $k$. The theoretical result (\textit{Model}) is obtained using Eq.~(\ref{Eq: Transmissions}). The simulation result (\textit{Simul}) is the average of the 10000 independent runs per configuration. We can first highlight that the difference between the two results is almost negligible (the maximum relative error is less than $0.8 \%$, showing the validity of the proposed model). Another interesting aspect that can be seen is that for larger densities ($\approx 1/2$, i.e $w = k/2$) the average number of additional transmissions is $\approx 1.6$, for all $k$; this result matches the value that would have been obtained for the traditional \emph{RLNC} approach~\cite{Ho2003}. However, for lower densities, i.e. $w=3$, the number of required transmissions increases considerably, especially for larger $k$.

\begin{table}[t]
\begin{center}
\caption{Average number of transmissions required to successfully decode a generation for different configurations}
\label{Tab: Transmissions}
\begin{tabular}{| c | c | c | c | c | c |}
\hline
\multicolumn{2}{|c|}{}       & $w=3$ 	& $w=7$  & $w=15$ & $w=31$ \\ \hline
\multirow{2}{*}{$k=32$}      & Model  & 43.83 	& 33.62  & 33.58  & -  \\ 
							 & Simul  & 44.17 	& 33.64  & 33.60  & - \\ \hline
\multirow{2}{*}{$k=64$}      & Model  & 100.34    & 65.92  & 65.62  & 65.62   \\
					         & Simul  & 101.49    & 65.91  & 65.61  & 65.60   \\ \hline
\multirow{2}{*}{$k=128$}     & Model & 230.36	& 131.22 & 129.85 & 129.85	\\
					         & Simul & 231.89	& 131.19 & 129.60 & 129.60	\\ \hline
\end{tabular}
\end{center}
\end{table}

In order to complement the average values collected in Table~\ref{Tab: Transmissions}, Figure~\ref{Fig: probabilityofDecoding} shows the probability that the successful decoding of the whole generation happened after receiving $\epsilon$ additional packets (i.e $k + \epsilon$). Again, the difference between the results obtained by means of simulation and those using the model is negligible, the mean squared error in the worst case ($w=3$ and $k=128$) is rather low, $=9.35\cdot 10^{-5}$. Note that lower densities would yield a worse performance, since the probability of successfully decoding a complete generation is very small when the number of additional transmissions is low. However, when $w$ increases, such probability equals almost 1 for just 5 additional transmissions. 

\begin{figure}[t]
	\centering
	\def \figurewidth {0.3\columnwidth}
	\def \figureheight {4.0cm}
	\subfloat[][$k=64$]{\pgfplotstableread{prob_decoding_N_k_64_w_3.dat}{\dataset}
\pgfplotstableread{prob_decoding_N_k_64_w_15.dat}{\datasetDos}
\begin{tikzpicture} 
font = \scriptsize, 
\begin{axis}[scale only axis, 
    width=\figurewidth,
    height=\figureheight,
    mark options={solid},
    ymin=0,
    ymax=1,
	xmin=1,
	xmax=100,
    ylabel={prob of decoding},
    xlabel={$\epsilon$},
    compat=1.3,
    legend pos= south east,
            legend cell align=left,
	legend style={draw=none, fill=none,legend columns=1, font = \footnotesize}
    ]

\addplot [color=black, thick, mark=o, only marks, mark repeat=8]
table[x index = 0, y index =1] from \dataset;
\addlegendentry{$w=3$}
\addplot [color=black, thick, mark=triangle, only marks, mark repeat=5]
table[x index = 0, y index =1] from \datasetDos;
\addlegendentry{$w=15$}

\addplot [color=black]
table[x index = 0, y index =2] from \dataset;

\addplot [color=black]
table[x index = 0, y index =2] from \datasetDos;

\end{axis}
\end{tikzpicture}}
	\subfloat[][$k=128$]{\pgfplotstableread{prueba3.dat}{\dataset}
\pgfplotstableread{prueba15.dat}{\datasetDos}
\begin{tikzpicture} 
font = \scriptsize, 
\begin{axis}[scale only axis, 
width=\figurewidth,
height=\figureheight,
mark options={solid},
ymin=0,
ymax=1,
xmin=1,
xmax=225,
ylabel={prob of decoding},
xlabel={$\epsilon$},
compat=1.3,
legend pos= south east,
legend cell align=left,
legend style={draw=none, fill=none,legend columns=1, font = \footnotesize}
]

\addplot [color=black, thick, mark=o, only marks, mark repeat=8]
table[x index = 0, y index =1] from \dataset;
\addlegendentry{$w=3$}
\addplot [color=black, thick, mark=triangle, only marks, mark repeat=5]
table[x index = 0, y index =1] from \datasetDos;
\addlegendentry{$w=15$}

\addplot [color=black]
table[x index = 0, y index =2] from \dataset;

\addplot [color=black]
table[x index = 0, y index =2] from \datasetDos;

\end{axis}
\end{tikzpicture}}

	\caption{Probability of a successful decoding event after $\epsilon$ extra transmissions. Markers show the values obtained with the proposed model and the line corresponds to the simulation results}
	\label{Fig: probabilityofDecoding}
\end{figure}
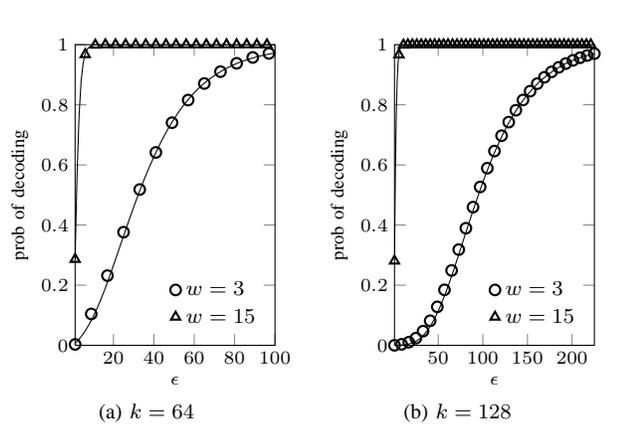

\begin{figure}[t]
	\centering
	\def \figurewidth {0.7\columnwidth}
	\def \figureheight {3.0cm}
	\subfloat[][Number of operations]{\label{Fig: complexity}\pgfplotstableread{complexity_operations_comparative.dat}{\dataset}

\begin{tikzpicture} 
font = \scriptsize, 
\begin{axis}[scale only axis, 
    width=\figurewidth,
    height=\figureheight,
    mark options={solid},
    ymin=0,
	xmin=3,
	xmax=64,
    ylabel={\# operations},
    xlabel={$r$},
    compat=1.3,
    legend pos= north west,
        legend cell align=left,
	legend style={draw=none, fill=none,legend columns=1, font = \footnotesize}
    ]

\addplot [color=black, mark=triangle*, thick, mark repeat=10, mark phase = 2]
plot[]
table[x index = 0, y expr=\thisrowno{1}*0.001] from \dataset;
\addlegendentry{\emph{TSNC-Model}}


\addplot [color=black, mark=*,mark size=1.5pt, thick, mark repeat=10, mark phase = 7]
plot[]
table[x index = 0, y expr=\thisrowno{4}*0.001] from \dataset;
\addlegendentry{\emph{TSNC-Bound}}

\addplot[color=black, dashed]
plot[]
table[x index = 0, y expr=\thisrowno{7}*0.001] from \dataset;
\addlegendentry{$w=3$}

\addplot [color=black, mark=square,mark size=1.7pt,thick, mark repeat=10]
plot[]
table[x index = 0, y expr=\thisrowno{10}*0.001] from \dataset;
\addlegendentry{\emph{RLNC}}

\end{axis}
\end{tikzpicture}} \\
	\subfloat[][Number of transmissions]{\label{Fig: transmissions}\pgfplotstableread{complexity_operations_comparative.dat}{\dataset}

\begin{tikzpicture} 
font = \scriptsize, 
\begin{axis}[scale only axis, 
    width=\figurewidth,
    height=\figureheight,
    mark options={solid},
    ymin=0,
	xmin=3,
	xmax=64,
    ylabel={\# transmissions},
    xlabel={$r$},
    compat=1.3,
    legend pos= north west,
        legend cell align=left,
	legend style={draw=none, fill=none,legend columns=1, font = \footnotesize}
    ]

\addplot [color=black, mark=triangle*, thick, mark repeat=10 ]
plot[]
table[x index = 0, y index = 2] from \dataset;
\addlegendentry{\emph{TSNC-Model}}


\addplot [color=black, mark=*, mark size=1.5pt, thick, mark repeat=10, mark phase = 7]
plot[]
table[x index = 0, y index = 5] from \dataset;
\addlegendentry{\emph{TSNC-Bound}}

\addplot[color=black, dashed]
plot[]
table[x index = 0, y index = 8] from \dataset;
\addlegendentry{$w=3$}

\addplot [color=black, mark=square,mark size=1.7pt,thick, mark repeat=10, mark phase= 0]
plot[]
table[x index = 0, y index = 11] from \dataset;
\addlegendentry{\emph{RLNC}}

\end{axis}
\end{tikzpicture}}
	\caption{Number of operations by the decoder and transmissions by the encoder when the decoder has already received $r$ linearly independent packets}
	\label{Fig: complexityAndTransmissions}
\end{figure}
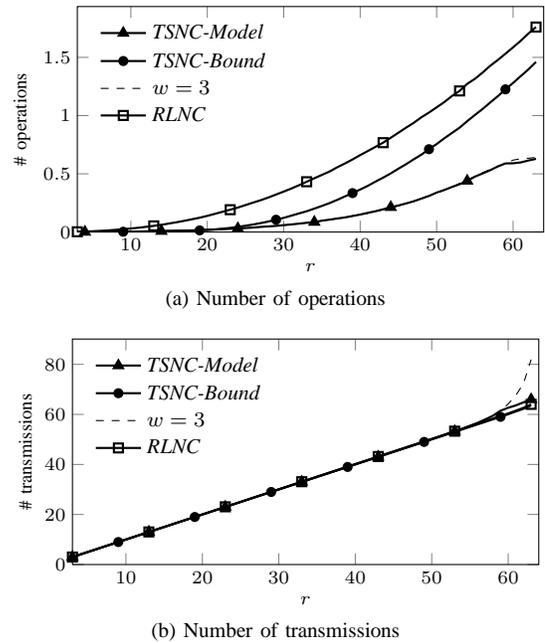

Up to now, we have argued that the use of low coding densities would be beneficial since they would yield shorter decoding times. In order to assess this, Figure~\ref{Fig: complexity} shows the number of operations that have been executed at the decoder from the first reception, until the packet is successfully decoded. On the other hand, Figure~\ref{Fig: transmissions} shows the average number of coded packets sent by the encoder.  We have used the KODO library, which logs the number of operations that are executed at the receiver. The complexity of an individual operation is alike for any $w$ value, so the results shown in the figure provide a rather precise idea of the complexity of the decoding process. We plot the number of accumulated operations and number of total transmissions versus the current rank at the receiver. The figures show the average value after 10000 independent runs. When a \emph{TSNC} approach is used, the encoder has perfect knowledge of the decoder state and changes the density when the expected number of transmissions to increase the rank by $1$ at the decoder is higher than $1.1$. By using the model proposed we have a higher reduction on complexity. In particular, the complexity achieved is almost the same using a fixed \ac{SNC} approach and $w=3$ but, note that in \ac{SNC} case the number of transmissions at the end is considerably high. We can clearly see that a dynamic tuning of the coding density would certainly yield relevant improvements. We could keep the density at a low value until such probability starts to decrease and then shift to a higher $w$.


After validating the proposed model over an ideal channel, we broaden the assessment to include error-prone scenarios. For that we introduce a certain loss probability, $\xi$, over the link between the source and the receiver, so that some packets are lost during their transmission. Figure~\ref{Fig: AverageErrors} shows how the average number of transmissions that are required to successfully decode a generation, with $k=64$, as a function of $\xi$. As can be seen there is again a very tight match between the simulation results and the values obtained with the proposed model, the relative error is less than $0.8 \%$. The impact of greater $\xi$ values is more visible for low densities, since the number of required transmissions would increase quite strongly. 

\begin{figure}[t]
	\centering
	\def \figurewidth {0.7\columnwidth}
	\def \figureheight {3.0cm}
	\pgfplotstableread{averageTx_FER.dat}{\dataset}
\begin{tikzpicture} 
font = \scriptsize, 
\begin{axis}[scale only axis, 
    width=\figurewidth,
    height=\figureheight,
    mark options={solid},
	xmin=0,
	xmax=0.8,
    ylabel={\# of packets transmitted},
    xlabel={$\xi$},
    compat=1.3,
    legend pos= north west,
	legend cell align=left,
	legend style={draw=none, fill=none,legend columns=1, font = \footnotesize}
    ]

\addplot [color=black,forget plot]
table[x index = 0, y index =3] from \dataset;

\addplot [color=black,forget plot]
table[x index = 0, y index =4] from \dataset;

\addplot [color=black, thick, mark = square, only marks]
table[x index = 0, y index =1] from \dataset;
\addlegendentry{$w=3$}

\addplot [mark=o, , thick, only marks]
table[x index = 0, y index =2] from \dataset;
\addlegendentry{$w=31$}

\end{axis}
\end{tikzpicture}
	\caption{Average number of transmissions required to decode a generation Vs. Packet Loss Rate ($\xi$). Markers show the values obtained with the proposed model and the line corresponds to the simulation results}
	\label{Fig: AverageErrors}
\end{figure}
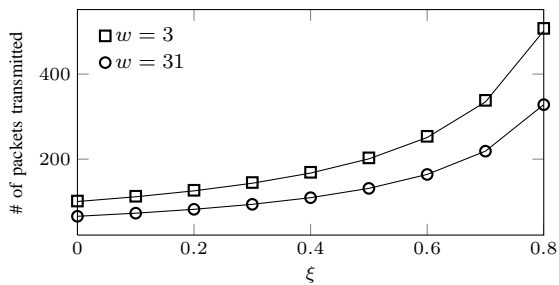

The previous result are broaden in Figure~\ref{Fig: probabilityofDecodingFER}, which shows the impact of packet erasure channels over the probability of being able to decode a generation after $N = k + \epsilon$ transmissions. We have used a low density, $w=3$, since it yields the worse behavior in terms of the additional transmissions that are required. The impact of the hostile links is again clearly seen in the figures. On the other hand, we can see that in this case there is a slight difference between the simulations results and the values obtained with the model. As was already discussed, the fitting for $w = 3$ was less accurate than the ones for larger densities; in any case, the differences are rather small, and the results clearly show that the proposed model can be also be used over error-prone links. In this case the mean squared error when $\xi=0.3$ is $3.90 \cdot 10^{-4}$.

\begin{figure}[t]
	\centering
	\def \figurewidth {0.7\columnwidth}
	\def \figureheight {3.0cm}
	\pgfplotstableread{prob_decoding_N_FER.dat}{\dataset}
\begin{tikzpicture} 
font = \scriptsize, 
\begin{axis}[scale only axis, 
    width=\figurewidth,
    height=\figureheight,
    mark options={solid},
    ymin=0,
    ymax=1,
	xmin=1,
	xmax=100,
    ylabel={prob of decoding},
    xlabel={$\epsilon$},
    compat=1.3,
    legend pos= north west,
	legend cell align=left,
        mark repeat={8},
	legend style={draw=none, fill=none,legend columns=1, font = \footnotesize}
    ]

\addplot [color=black,forget plot]
table[x index = 0, y index =4] from \dataset;

\addplot [color=black,forget plot]
table[x index = 0, y index =5] from \dataset;

\addplot [color=black,forget plot]
table[x index = 0, y index =6] from \dataset;
	  
\addplot [color=black, thick,mark=square, only marks]
table[x index = 0, y index =1] from \dataset;
\addlegendentry{$\xi=0.0$}

\addplot [color=black, thick, mark=o, only marks]
table[x index = 0, y index =2] from \dataset;
\addlegendentry{$\xi=0.1$}

\addplot [color=black, thick, mark= triangle, only marks]
table[x index = 0, y index =3] from \dataset;
\addlegendentry{$\xi=0.3$}

\end{axis}
\end{tikzpicture}
	\caption{Probability of a successful decoding event after $\epsilon$ extra transmissions, $k=64$ and $q=1$. Markers show the values obtained with the proposed model and the line corresponds to the simulation results}
	\label{Fig: probabilityofDecodingFER}
\end{figure}
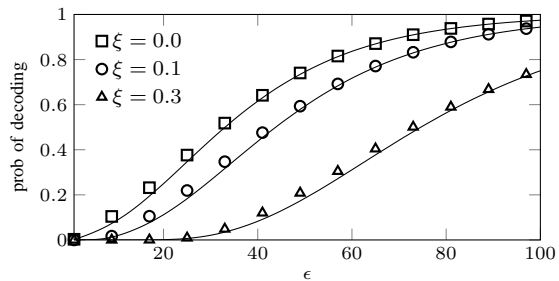

\section{Conclusions and future work}
\label{Sec: Conclusions}

In this paper we have presented the first complete model to mimic the behavior of sparse coding techniques. It is based on an Absorbing Markov Chain, where the states are defined as the combination of the rank and the non-zero columns of the corresponding decoding matrix. After finding the corresponding transition probabilities, we have exploited the properties of the Absorbing Markov Chains to derive some some key performance figures: average number of transmissions, probability of successfully decoding a generation after $N$ transmissions, and probability of receiving a linearly independent packet. We have also seen that the model can be broadened to consider error-prone links. Afterwards, the model has been validated by means an extensive simulation, yielding an almost perfect match and clearly outperforming some of the bounds that have been used in previous works.

This model can be easily exploited to select the most appropriate density in sparse coding techniques. In particular, in a \ac{TSNC} solution, the encoder decides the optimum density value as the transmission evolves, for instance taking into account the trade-off between useless transmissions (linear dependent combinations) and the reduction on computational complexity. Hence, the proposed model could provide a better insight to establish optimum configurations of \ac{TSNC} solutions.

%
%

\begin{acronym}

  \acro{NC}{Network Coding}
  \acro{RLNC}{Random Linear Network Coding}
  \acro{RTT}{Round Trip Time}
  \acro{RTO}{Retransmission TimeOut}
  \acro{WMN}{Wireless Mesh Network}
  \acro{WMNs}{Wireless Mesh Networks}
  \acro{MSS}{Maximum Segment Size}
  \acro{MAC}{Medium Access Control}
  \acro{SDU}{Service Data Unit}
  \acro{PDU}{Protocol Data Unit}
  \acro{DIFS}{Distributed Inter-Frame Space}
  \acro{MTU}{Maximum Transfer Unit}
  \acro{ISP}{Internet Service Provider}
  \acro{RLC}{Random Linear Coding}
  \acro{HMP}{Hidden Markov Process}

  \acro{FER}{Frame Error Rate}
  \acro{PER}{Packet Error Rate}

  \acro{PLCP}{Physical Layer Convergence Protocol}

  \acro{ARQ}{Automatic Repeat Request}
  \acro{FEC}{Forward Error Correction}

  \acro{RLSC}{Random Linear Source Coding}
  \acro{RLNC}{Random Linear Network Coding}
  \acro{MORE}{MAC-independent Opportunistic Routing \& Encoding}
  \acro{CATWOMAN}{Coding Applied To Wireless On Mobile Ad-Hoc Networks}
  \acro{CTCP}{Network Coded TCP}
  
  \acro{cdf}{cumulative distribution function}
  
  \acro{BATS}{BATched Sparse}
  \acro{TSNC}{Tunable Sparse Network Coding}
  \acro{SNC}{Sparse Network Coding}

\end{acronym}


\section*{Acknowledgements}

This work has been supported by the Spanish Government (Ministerio de Econom\'ia y Competitividad, Fondo Europeo de Desarrollo Regional, FEDER) by means of the projects \textbf{COSAIF}, \emph{``Connectivity as a Service: Access for the Internet of the Future''} (TEC2012-38754-C02-01), and \textbf{ADVICE} (TEC2015-71329-C2-1-R).  This work was also financed in part by the \textbf{TuneSCode} project (No. DFF 1335-00125) granted by the Danish Council for Independent Research.

\bibliographystyle{IEEEtran}
\bibliography{biblio}

\end{document}